\documentclass[pra,aps,superscriptaddress,twocolumn,nopacs, nofootinbib]{revtex4-1}
\usepackage{graphicx,color,epstopdf}
\usepackage{amsmath,amsfonts,enumerate,amsthm,amssymb,mathtools}
\usepackage{graphicx,color,epstopdf}

\usepackage{enumitem}
\usepackage{thmtools,thm-restate}
\usepackage{bbold}
\usepackage{hyperref}
\usepackage{multirow}
\usepackage{subfigure} 
\usepackage{mathdots} 

\hypersetup{
   colorlinks=true,         
   linkcolor=blue,          
   citecolor=blue           
}

\def\R{ \mathbb{R} }

\def\>{\rangle}
\def\<{\langle}
\newcommand{\bra}[1]{\langle {#1} |}
\newcommand{\ket}[1]{| {#1} \rangle}
 
\newcommand{\ketbra}[2]{\ensuremath{\left|#1\right\rangle\!\!\left\langle#2\right|}}

\newcommand{\tr}[2]{\mathrm{Tr}_{#1}\left( #2 \right)}
\newcommand{\iden}{\mathbb{1}}

\newcommand{\be}{\begin{equation}}
\newcommand{\ee}{\end{equation}}

\definecolor{nicepurple}{rgb}{0.3,0,0.6}

\def\Re{\operatorname{Re}}
\usepackage{pifont}

\theoremstyle{plain}
\newtheorem{thm}{Theorem}
\newtheorem{lem}[thm]{Lemma}
\newtheorem{corol}[thm]{Corollary}
\theoremstyle{definition}
\newtheorem{defn}{Definition}
\theoremstyle{remark}

\begin{document}

\title{Quantum fluctuation theorems, contextuality and work quasi-probabilities}
\begin{abstract}

We discuss the role of contextuality within quantum fluctuation theorems, in the light of a recent no-go result by Perarnau \emph{et al}. We show that any fluctuation theorem reproducing the two-point-measurement scheme for classical states either admits a notion of work quasi-probability or fails to describe protocols exhibiting contextuality. Conversely, we describe a protocol that smoothly interpolates between the two-point measurement work distribution for projective measurements and Allahverdyan's work quasi-probability for weak measurements, and show that the negativity of the latter is a direct signature of contextuality.
\end{abstract}
\author{Matteo Lostaglio}
\affiliation{ICFO-Institut de Ciencies Fotoniques, The Barcelona Institute of Science and Technology, Castelldefels (Barcelona), 08860, Spain}


\maketitle

While \emph{quantum} thermodynamics thrived in recent years \cite{goold2016role, vinjanampathy2016quantum}, we still lack clear evidence that there are any thermodynamically relevant protocols whose results cannot, in a precise sense, be ``simulated'' classically. In this work we show that such protocols do indeed exist; the solution lies in a long-standing debate that surrounded the definition of work in so-called fluctuation theorems (FTs).

FTs are one of the most important set of results in non-equilibrium thermodynamics \cite{jarzynski1997nonequilibrium,crooks1999entropy, sevick2008fluctuation, campisi2011colloquium}. In their simplest form, a classical system initially prepared in a thermal state at temperature $T$ is driven out of equilibrium. This is achieved by changing the parameters of the Hamiltonian from $H(0)$ to $H(\tau)$, according to a fixed protocol, while keeping the system isolated from the environment. Each repetition requires an amount of work $w$, corresponding to the realisations of a random variable $W$. Denote by $Z_{H(t)}$ the partition function of $H(t)$ and by $k$ Boltzmann's constant. The free energy difference between the equilibrium state with respect to the final Hamiltonian and the initial equilibrium state reads $\Delta F = -kT \log Z_{H(\tau)}/Z_{H(0)}$. Then the Jarzynski equality characterises work fluctuations in the protocol above~\cite{jarzynski1997nonequilibrium}:
\begin{equation}
\label{eq:jar}
\langle e^{-\beta W} \rangle = e^{-\beta \Delta F},
\end{equation}
where $\langle \cdot \rangle$ denotes averaging and $\beta =(kT)^{-1}$. Through Jensen's inequality, Eq.~\eqref{eq:jar} can be seen as a generalisation of the standard thermodynamic inequality \mbox{$\langle W \rangle \geq \Delta F$}. However, it also encodes information about fluctuations, e.g. the probability that \mbox{$W> \Delta F + x$} is bounded by $e^{-\beta x}$. Eq.~\eqref{eq:jar} and related equalities provide refined statements of the second law of thermodynamics beyond averages, are valid at the microscopic scale far from the thermodynamic limit, recover results from linear response theory and give a way to measure the free energy from non-equilibrium measurements of work \cite{jarzynski2011equalities}.   

Much effort has been devoted to finding analogous results characterising work fluctuations of quantum processes and in the presence of quantum coherence  \cite{allahverdyan2005fluctuations, campisi2011colloquium, albash2013fluctuation,rastegin2014jarzynski, allahverdyan2014nonequilibrium, solinas2015full,manzano2015nonequilibrium, aberg2016fully, alhambra2016fluctuating, perarnau2016no, miller2016time}.  One of the main challenges is the definition of work for quantum systems (see \cite{jarzynski2015quantum} and references therein). Work in closed systems is defined, in classical physics, as the energy difference between the initial and final phase space point, while in the quantum setting the conventional approach adopts the \emph{two-point-measurement} (TPM) scheme \cite{talkner2007fluctuation}. The idea is to 
define work as the energy difference between the outcomes of two projective measurements of energy, performed at the beginning and at the end of the protocol. 

Whilst the TPM definition of work has become standard \cite{campisi2011colloquium}, various authors claimed the approach has some important limitations \cite{allahverdyan2005fluctuations}. Firstly, some observed that, projective measurements being invasive, the average work according to this definition does \emph{not} in general coincide with the average energy change if no measurement is performed \cite{allahverdyan2014nonequilibrium, solinas2015full, miller2016time}. Secondly, others pointed to the fact that quantum coherence and entanglement are destroyed by the TPM scheme at the start of the protocol \cite{solinas2015full, solinas2016probing, lostaglio2015description, aberg2016fully}.

Concerning the first of the two issues raised, a recent no-go result has shown that it is a universal feature, proving that within quantum theory no FT can simultaneously reproduce the TPM scheme for classical states and respect the identification of average work with average energy change~\cite{perarnau2016no}. However, it is unclear if this identification is indeed a property we should impose on quantum FTs. Regarding the second issue, one can argue that the evolution described by $H(t)$ generates quantum coherence again and, in fact, interference effects do appear in the TPM work distribution when compared to the classical limit \cite{jarzynski2015quantum}. Hence, what are the fundamental limitations, if any, of current FTs?

There is a caveat in the no-go theorem of Ref.~\cite{perarnau2016no}, in that it can be circumvented at the price of extending the work distributions to quasi-probabilities \cite{allahverdyan2014nonequilibrium, solinas2015full,solinas2016probing,  miller2016time}. The occurrence of negative values in such distributions has been considered a limitation by some authors \cite{allahverdyan2014nonequilibrium, talkner2016aspects}, while others claimed they signal quantum effects \cite{solinas2015full, miller2016time}, since they can be related to the violation of the Leggett-Garg inequality \cite{leggett1985quantum, clerk2011full}. A second natural question is then: to what extent quasi-probabilities are a \emph{necessary} ingredient to capture quantum effects in FTs and what exactly can we infer from observing their negativity? 

Here we contribute to these issues by showing that
\begin{enumerate}
	\item Any FT that reproduces the two-point measurement scheme for classical states is either based on work quasi-probabilities or admits a non-contextual description. In other words, restricting to work probabilities prevents us from probing stronger forms of non-classicality.
	\item Conversely, in a generalisation of the TPM scheme we show that the appearance of negativity in a work quasi-probability signals the onset of contextuality.
\end{enumerate} 

\subsection{Setting the scene: A no-go result for fluctuation theorems}

A system prepared in a state $\rho$ undergoes a unitary evolution $U$ between time $0$ and $\tau$, induced by a time-dependent Hamiltonian $H(t)$. By comparison to the classical case, one wishes to define a distribution $p(w|\mathcal{P})$ giving the probability that the protocol $\mathcal{P}$ requires an amount of work $w$ to be realised. Classically, $w$ can be defined as the internal energy change of the system. Quantum mechanically, however, the mere act of probing the initial state $\rho$ will induce a disturbance $\rho \mapsto \sigma$. We will then consider a general protocol $\mathcal{P}$ schematically described as
\begin{equation}
\label{eq:protocols}
	 \mathcal{P}:= \{ (H(0), \rho) \longmapsto (H(\tau),U \sigma U^\dag)\},
\end{equation} where $U = \mathcal{T} \exp \left(-i \int_{0}^\tau dt H(t) \right)$ and $\mathcal{T}$ is the time-ordering operator. To fix the notation,
\begin{equation*}
H(0)=\sum_i E_i \ketbra{i}{i}:=\sum_i E_i \mathcal{E}_i , \quad H(\tau)=\sum_i E'_i \ketbra{i'}{i'}.
\end{equation*}

The standard framework extracts the work statistics from the TPM scheme, measuring $H(0)$ at the beginning of the protocol and $H(\tau)$ at the end \cite{talkner2007fluctuation}. Upon observing, respectively, outcomes $i$ and $j$, one sets \mbox{$w = E'_j - E_i$}. The random variable $W$ defined in this way satisfies Eq.~\eqref{eq:jar} \cite{campisi2011colloquium}. 

Note, however, that the first measurement can modify the subsequent statistics. This becomes evident looking at the problem in Heisenberg picture, where one can think of the TPM scheme as the sequential measurement of $H(0)$ followed by $U^\dag H(\tau) U$. Whenever $[H(0),U^\dag H(\tau) U] \neq 0$ and $[\rho,H(0)]\neq 0$, the statistics of the second measurement will be disturbed by the first. In fact, $\sigma = \mathcal{D}_{H(0)}(\rho)$, with $\mathcal{D}_{H(0)}$ denoting the operation that fully dephases in the eigenbasis of $H(0)$.

A recent result formalised this clash into a no-go theorem \cite{perarnau2016no}. Consider the definition:
\begin{defn}
	A protocol $\mathcal{P}$ is called a \emph{FT protocol} if for initial states with no quantum coherence, i.e., $\rho$ satisfying $[\rho, H(0)] = 0$, the results of the TPM scheme are recovered:
	\begin{equation}
	\nonumber
	p(w|\mathcal{P}) = p_{\rm tpm}(w|\mathcal{P}):=\sum_{E'_j -E_i =w} p_i p_{j|i},
	\end{equation}
	where $p_i = \bra{i}\rho\ket{i}$, $p_{j|i} = |\bra{j'}U\ket{i}|^2$.
\end{defn}
In other words, FT protocols are those that recover the TPM scheme at least in those situations in which the measurement does not introduce any disturbance to the evolution of the system. This is sufficient to reproduce Eq.~\eqref{eq:jar} for thermal initial states and to match the expected definition of work in the classical limit (in the cases analysed in \cite{jarzynski2015quantum}). Consider now the following assumptions:
\begin{enumerate}
	\item (Work distribution) \label{ass0} $\mathcal{P}$ measures a work probability  distribution $p(w|\mathcal{P})$ convex under mixtures of protocols. More precisely, let  $q \in [0,1]$ and $\mathcal{P}^i$ be protocols only differing by the initial preparation $\rho_i$. One requires that if \mbox{$\rho_0 = q \rho_1 + (1-q) \rho_2$} then
	\begin{equation}
	\nonumber
	p(w|\mathcal{P}^0) = q	p(w|\mathcal{P}^1) + (1-q) 	p(w|\mathcal{P}^2).
	\end{equation}
			\item (Average work) \label{ass2} The average \emph{measured} work should reproduce the average energy change induced by the unitary process on the \emph{initial state} 
		\begin{equation}
		\nonumber
		\langle W \rangle := \sum_w p(w|\mathcal{P}) w = \tr{}{U \rho U^\dag H(\tau)} - \tr{}{\rho H(0)}.
		\end{equation}
	\end{enumerate}
	Assumption~\ref{ass0} includes the natural demand that, upon conditioning the choice of the protocol on a coin toss, the measured fluctuations are simply the convex combination of those observed in the individual protocols. Assumption~\ref{ass2} is based on the identification of average work with average energy change in closed systems. The main result of Ref.~\cite{perarnau2016no} is that no FT protocol can satisfy both \ref{ass0} and \ref{ass2} when the system has no matching gaps, i.e., $E'_{j_1} - E_{i_1} \neq E'_{j_2} - E_{i_2} $ if $(j_1,i_1)\neq (j_2,i_2)$.

\subsection{Genuinely non-classical effects}

While the result of Ref.~\cite{perarnau2016no} was phrased mainly as an incompatibility between the requirement that $\mathcal{P}$ is a FT protocol and assumption \ref{ass2}, we note that assumption~\ref{ass0} contains the often implicit condition that a work distribution \emph{exists}. Since work involves two generally non-commuting observables $H(0)$ and $U^\dag H(\tau) U$, restrictions arise from the lack of a joint probability distribution for them. 

The above point can be sharpened as follows. Broadly speaking, we want to make a distinction between a phenomenon that is essentially classical in nature from one that is irreducibly quantum-mechanical \cite{jennings2016no}. We may think in terms of a challenge involving two parties, Alice and Bob. Alice sets up a quantum experiment, specifying the quantum systems involved, their interactions and the measurements performed, and tells Bob about this. Bob then prepares a box, in which he devises some classical mechanism that tries to reproduce the same statistics of Alice's experiment. If -- despite Bob's best efforts -- Alice can always find discrepancies between the statistics produced by Bob and her quantum experiment, we call the quantum phenomenon involved \emph{genuinely non-classical}.

The idea of classical mechanism is encapsulated in the notion of hidden variable model or \emph{ontological model}.
At the operational level, consider a set of instructions defining preparations procedures $P$ and measurement procedures $M$ with outcomes $k$; Alice observes $k$ with probability $p(k|P,M)$. Bob's classical mechanism reproduces this statistics using a set of states $\lambda$ that are, in general, randomly sampled from a set $\Lambda$ according to a probability distribution $p(\lambda|P)$ every time the preparation $P$ is performed. For example, $\Lambda$ may be the phase space of the classical mechanism, but we are not limited to that. Moreover, Bob's mechanism can include a measurement device $M$ that takes in the physical state $\lambda$ and outputs an outcome $k$ with probability $p(k|\lambda,M)$. Bob wins the challenge if he can reproduce the statistics $p(k|P,M)$ as an average over the unobserved states of the classical mechanism~\cite{spekkens2005contextuality, kunjwal2015kochen}:
\begin{equation}
	\label{eq:ontologicalreproduce}
	p(k|P,M) = \int_\Lambda d\lambda p(\lambda|P) p(k|\lambda,M).
\end{equation}

For a mechanism to be classical it should operate in a non-contextual way. Specifically, the mechanism is called \emph{preparation non-contextual} if $p(\lambda|P)$ is a function of the quantum state alone, i.e. \mbox{$p(\lambda|P) \equiv p(\lambda|\rho)$}; for example, Bob's mechanism cannot distinguish different ensembles associated to the same $\rho$. Furthermore, the mechanism is called \emph{measurement non-contextual} if $p(k|\lambda,M)$ depends only on the POVM element $M_k$ associated to the corresponding outcome of the measurement $M$, i.e. \mbox{$p(k|\lambda,M) \equiv p(k|\lambda,M_k)$} \cite{spekkens2005contextuality, spekkens2014status}. If a mechanism is both preparation and measurement non-contextual, it is called \emph{universally non-contextual}. See Supplemental Material~A for the relation to Kochen-Specker contextuality \cite{kochen1975problem}.

It is a question of fundamental as well as practical importance to know if FT protocols can uncover genuinely non-classical phenomena. Here we make this precise by showing:
\begin{thm}
	\label{thm:nc}
	Assume the FT protocol $\mathcal{P}$ satisfies assumption~\ref{ass0}. Then there exists a universally non-contextual ontological model for every preparation $\rho$ and measurement of $p(w|\mathcal{P})$.
\end{thm}
Theorem~\ref{thm:nc} says that in any FT protocol we either 
\begin{enumerate}
	\item[(a)] Lift the assumption that the work distribution should be a probability (assumption \ref{ass0}).
	\item[(b)] Only probe quantum effects admitting a classical non-contextual model.
\end{enumerate}  

\subsection{Existence of a work distribution forces non-contextuality}
\label{sec:proof1}

Let us prove Theorem~\ref{thm:nc}. From now on, we focus on the case in which for every $w$ there exists a unique couple of indexes $i,j$ such that \mbox{$E'_j -E_i = w$}. Let $P$ and $M$ denote the preparation and measurement procedures involved in the FT protocol $\mathcal{P}$. As shown in Ref.~\cite{perarnau2016no}, assumption~\ref{ass0} can be reformulated as follows: there exists a POVM $M(\mathcal{P}) = \{M_w(\mathcal{P})\}$, such that
\begin{equation}
\label{eq:povmwork}
p(w|\mathcal{P}) = \tr{}{M_w(\mathcal{P}) \rho},
\end{equation}
with $M_w(\mathcal{P})$ being a function of $H(0)$, $U$ and $H(\tau)$, but \emph{not} of $\rho$.

We want to derive the existence of a non-contextual model reproducing the observed work statistics $p(w|\mathcal{P})$, i.e., from Eq.~\eqref{eq:ontologicalreproduce},
\begin{equation}
\label{eq:task}
\tr{}{\rho M_w(\mathcal{P})} = \int_\Lambda d\lambda p(\lambda|\rho) p(w|\lambda,M_w(\mathcal{P})).
\end{equation}

This, in fact, arises as a simple consequence of the main results of Ref.~\cite{perarnau2016no}. There it is shown that assumption~\ref{ass0} enforces on the FT protocol \mbox{$M_w(\mathcal{P}) = M^{\rm tpm}_w(\mathcal{P})$}, where $M^{\rm tpm}_w(\mathcal{P})$ is the two-point-measurement POVM~\cite{roncaglia2014work}, $M^{\rm tpm}_w(\mathcal{P}) = |\bra{j'}U\ket{i}|^2 \ketbra{i}{i}$. $\Lambda$ can then be taken to be a space labelling energies and for all $H(0)$, $U$ and $H(\tau)$,
	\begin{equation}
\label{eq:ontologicalmodel}
\begin{split}
p(\lambda|\rho) &= \bra{\lambda}\rho \ket{\lambda}, \\
p(w|\lambda,M_w(\mathcal{P})) &= \tr{}{M^{\rm tpm}_w(\mathcal{P}) \ketbra{\lambda}{\lambda}}.
\end{split}
	\end{equation}
Substituting Eqs.~\eqref{eq:ontologicalmodel} in Eq.~\eqref{eq:task}, and using Eq.~\eqref{eq:povmwork} we get the claimed result.

As it is known, there are prepare and measure experiments that \emph{cannot} be reproduced by a non-contextual mechanism \cite{spekkens2005contextuality, ferrie2011quasi}. While Bob can in principle emulate core aspects of many quantum phenomena with a classical mechanism \cite{spekkens2007evidence, bartlett2012reconstruction, spekkens2016quasi}, contextuality is beyond his reach. Theorem~\ref{thm:nc} says that any FT protocol satisfying assumption~\ref{ass0} will not allow the correspondent experiments to manifest genuine non-classicality, i.e. it will restricts us to probing a ``fragment'' of quantum theory \cite{leifer2014quantum, jennings2016no} that admits a classical representation \cite{ferrie2011quasi}. 

Some clarifications are now in order. First, when there is no  POVM satisfying Eq.~\eqref{eq:povmwork} (i.e., assumption~\ref{ass0} is lifted), 
we study if non-contextual ontological models exist for the scheme collecting the statistics through which  $p(w|\mathcal{P})$ is \emph{reconstructed}. Lifting assumption~\ref{ass0} expands the set of FT protocols $\mathcal{P}$ beyond the TPM POVM; but this is distinct from proving that any of the reconstruction protocols lacks a non-contextual mechanism. 

Second, Ref.~\cite{spekkens2008negativity} highlighted that the absence of non-negative quasi-probability \emph{representations} for a given protocol and contextuality are the same concept. However, one should be careful not to identify the negativity of $p(w|\mathcal{P})$, reconstructed from a set of preparations and measurements in~$\mathcal{P}$, with the negativity of \emph{every} quasi-probability representations of such preparations and measurements ($p(w|\mathcal{P})$ is not a representation, see Supplemental Material~C). It is a non-trivial task to define FT protocols in which negativity of a work quasi-probability can be provably associated to contextuality. This is what we will do in the rest of this work.

\subsection{Negativity of the work distribution implies contextuality}

Since we are interested in FT protocols able to witness genuinely non-classical features, due to Theorem~\ref{thm:nc} we investigate here the possibility of lifting assumption~\ref{ass0}. This means that $p(w|\mathcal{P})$ is a quasi-probability, exhibiting negativity (as in Ref.~\cite{allahverdyan2014nonequilibrium}) or lacking convexity (as in Ref.~\cite{sampaio2017impossible}). Here we investigate the former possibility. 

Firstly, let us describe a family of protocols that smoothly interpolates between the TPM protocol, estimating $p_{ \rm tpm}(w|\mathcal{P})$, and a protocol probing the recent work quasi-probability introduced by Allahverdyan in Ref.~\cite{allahverdyan2014nonequilibrium}, denoted by $p_{ \rm weak}(w|\mathcal{P})$. While the former is obtained in the strong (projective) measurement limit, the latter is achieved through a \emph{weak measurement} \cite{aharonov1988result, dressel2014colloquium}. Secondly, we will show that the negativity of $p_{ \rm weak}(w|\mathcal{P})$ directly signals contextuality of the weak measurement protocol. 

Consider the following one-parameter family of protocols, parametrized by $s \in \R$, involving these steps:\footnote{For related schemes, see Refs.~\cite{solinas2016probing, talkner2016aspects, hofer2017quasi} and references therein.}
 \begin{enumerate}
	\item  A measurement device or ``pointer'', represented by a one-dimensional quantum system with canonical observables $X$ and $P$, is prepared in a Gaussian state with spread $s$:
	\begin{equation}
	\label{eq:pointer}
		\ket{\Psi} = (\pi s^2)^{-1/4} \int dx \exp \left(-\frac{x^2}{2s^2}\right) \ket{x}.
	\end{equation}
	The system is prepared in state $\rho$, initially uncorrelated from the device.
	\item The device is coupled to the system through the interaction Hamiltonian $H_{ {\rm int}} = g(t) \mathcal{E}_i \otimes P$ over a time interval $[-t_M,0]$ (recall $\mathcal{E}_i = \ketbra{i}{i}$). We can choose units such that $g = \int_{-t_M}^{0} dt g(t) =1$. After the interaction, a projective measurement $\{\ketbra{x}{x}\}$ on the device induces a corresponding POVM $\{M^s_x\}_{x \in \R}$ on the system. In the limit $s \rightarrow \infty$, this is called a \emph{weak measurement} of $\mathcal{E}_i$.
	\item At $t=0$, the system is evolved according to $U$, the driving unitary of Eq.~\eqref{eq:protocols} (we neglect the free evolution of the system during the measurement).
	\item Finally, at $t=\tau$, a projective measurement of $H(\tau)$ is performed on the system and outcome $j$ is postselected.
\end{enumerate}

Denote by $q_j$ the probability of observing outcome $j$ in the final measurement. Moreover, let $\langle X \rangle_{j}$ be the expectation value of the shift in the pointer upon postselecting outcome $j$. There are two important limits (see Supplemental Material~B):
\begin{enumerate}
	\item In the strong measurement limit, $s \rightarrow 0$, $\{M^s_x\}$ approaches a projective measurement and
	\begin{equation*}
	q_j \langle X \rangle_{j} \rightarrow p_{\rm tpm}(w|\mathcal{P})= p_i p_{i|j}.
	\end{equation*} 
	\item In the weak measurement limit, $s \rightarrow \infty$, 
	\begin{equation*}
	q_j \langle X \rangle_{j} \rightarrow p_{\rm weak}(w|\mathcal{P})= \Re \tr{}{\rho \mathcal{E}_i \Pi_j }.
	\end{equation*}
\end{enumerate}
Here we defined $\Pi_j = U^\dag \ketbra{j'}{j'} U$ and $p_{\rm weak}(w|\mathcal{P})$ corresponds to the work quasi-distribution recently introduced by Allahverdyan \cite{allahverdyan2014nonequilibrium}:
\begin{equation}
\label{eq:allahberdyan}
p_{\rm{weak}}(w|\mathcal{P}) = \Re \tr{}{\rho \mathcal{E}_i \Pi_j}, \quad w= E'_j- E_i.
\end{equation}
  We see that both $p_{\rm tpm}(w|\mathcal{P})$ and $p_{\rm weak}(w|\mathcal{P})$ are obtained within the same general class of protocols. $p_{\rm weak}$ is known as the Margenau-Hill distribution \cite{margenau1961correlation}. The weak measurement protocol so defined is a FT protocol and we notice in passing that it satisfies assumption~\ref{ass2}. Furthermore, it gives rise to a FT~\cite{allahverdyan2014nonequilibrium}:
  \begin{equation}
  \label{eq:allahverdyanequality}
  \langle e^{-\beta W} \rangle = e^{-\beta \Delta F}\Upsilon,
  \end{equation}
  where $\Upsilon = \Re \tr{}{U^\dag\gamma(\tau)U \gamma(0)^{-1} \rho}$ and $\gamma(t)= e^{-\beta H(t)}/\tr{}{e^{-\beta H(t)}}$. If $\rho = \gamma(0)$, we have $\Upsilon=1$ and the equality of Eq.~\eqref{eq:jar} is recovered (for a full interpretation of Eq.~\eqref{eq:allahverdyanequality}, see Ref.~\cite{allahverdyan2014nonequilibrium}).

While $p_{\rm{weak}}$ can attain negative values, contrary to Ref.~\cite{allahverdyan2014nonequilibrium} we suggest this is not a limitation. Quite the opposite, due to Theorem~\ref{thm:nc} negativity is necessary to probe contextuality ($p_{\rm{weak}}$ is convex in $\rho$). Remarkably, it is also \emph{sufficient}. Recall that a classical mechanism is called \emph{outcome deterministic} if projective measurements $M$ give a deterministic response, i.e., $p(k|\lambda, M) \in \{0,1\}$. With reference to the protocols provided, we have:
\begin{thm}
	\label{thm:contextuality}
	Let $\rho$ be a quantum state, $\mathcal{E}_i$ and $\Pi_j$ projectors. If $\Re \tr{}{\rho \mathcal{E}_i \Pi_j} <0$, for $s$ large enough there is no measurement non-contextual ontological model for preparation $\rho$, measurement $\{M^s_x\}_{x \in \R}$ and post-selection $\Pi_j$ that satisfies outcome determinism.   
\end{thm}

The complete proof is given in Supplemental Material~D. The quantity $\Re \tr{}{\rho \mathcal{E}_i \Pi_j}/\tr{}{\rho \Pi_j}$ is known as \emph{generalised weak value} \cite{wiseman2002weak} and its negative values are called \emph{anomalous}. Hence, the above theorem is an extension of the main result of Ref.~\cite{pusey2014anomalous} to mixed states.\footnote{One may define generalised weak values as $\tr{}{\rho \mathcal{E}_i \Pi_j}/\tr{}{\rho \Pi_j}$, since this is a direct generalisation of the initial proposal~\cite{aharonovsusskind,dressel2010contextual}. However, complex weak values can be achieved in Gaussian quantum mechanics, which admits a non-contextual model \cite{karanjai2015weak, bartlett2012reconstruction}. Hence, following Ref.~\cite{pusey2014anomalous}, we focused on the real part.} The result implies that the observation of a generalised anomalous weak value provides a proof of contextuality of the FT protocol introduced above (note that $U$ is included in~$\Pi_j$).

We need to be more precise here, since the theorem involves the condition of outcome determinism. First note that this condition is indeed necessary to get any result: a measurement non-contextual (but not outcome deterministic) model exists for full quantum mechanics~\cite{spekkens2005contextuality}, and hence for any FT protocol. Second, note that many authors include outcome determinism in the definition of contextuality. This is the case of Kochen and Specker theorem \cite{kochen1975problem}. Third, the following corollary of Theorem~\ref{thm:nc} holds:
\begin{corol}
	Assume quantum mechanics holds. If $\Re \tr{}{\rho \mathcal{E}_i \Pi_j} <0$, for $s$ large enough there is no universally non-contextual model for preparation $\rho$, measurement $\{M^s_x\}_{x \in \R}$ and post-selection $\Pi_j$.
\end{corol}
This corollary holds because some elementary operational conditions (obviously satisfied by the operational theory associated to quantum mechanics) are sufficient to prove outcome determinism from preparation non-contextuality, as detailed in Refs.~\cite{spekkens2005contextuality, spekkens2014status}. 

A consequence of Theorem~\ref{thm:contextuality} is that in the presence of non-commutativity there is always a state able to witness contextuality in the FT protocol given above: for any $i,j$, with $[\mathcal{E}_i,\Pi_j]\neq 0$, there are quantum states $\rho$ such that $p_{\rm weak}(w|\mathcal{P})<0$ \cite{allahverdyan2014nonequilibrium}. Conversely, necessarily one must have $[\rho,\mathcal{E}_i] \neq 0$ and $[\rho,\Pi_j] \neq 0$ for some $i,j$ to observe negative values of the work distribution.

\subsection{Conclusions}

In this paper we presented the first example of contextuality in a thermodynamic framework. The no-go result of Theorem~\ref{thm:nc} shows that FT protocols are unable to access contextuality, \emph{unless} the notion of work distribution is extended to a work quasi-probability, lacking non-negativity or convexity. Conversely, from Theorem~\ref{thm:contextuality} we have seen that the negative values of a work quasi-probability (accessible through weak measurements) imply contextuality of a FT protocol that naturally generalises the TPM scheme. 

Importantly, since non-contextual models exist reproducing phenomena such as quantum interference, complementarity and non-commutativity -- among others \cite{spekkens2007evidence, bartlett2012reconstruction} --  contextuality cannot be understood as a consequence of measurement disturbance and lack of knowledge of an underlying classical variable. Hence, the negativity of $p_{\rm weak}$ is inherently non-classical (differentiating these protocols from TPM schemes). While a FT exists for $p_{\rm weak}$ (Eq.~\eqref{eq:allahverdyanequality}), it will be important to derive a \emph{direct} thermodynamic interpretation of negativity.

 This work also paves the way to an experimental verification of contextuality in a FT protocol, but more work needs to be done to make the present proposal robust to experimental imperfections~\cite{kunjwal2015kochen, mazurek2016experimental}.
 
 Analogous questions are being investigated in the context of quantum computing \cite{veitch2012negative, howard2014contextuality, veitch2014resource, delfosse2015wigner, bermejo2016contextuality}, and this work suggests a path to finding quantum advantages in thermodynamics. Eventually we hope this will lead to the design of thermodynamic machines, such as engines, whose performance provably outperform classical counterparts, independently of the specific assumptions on the underlying model.

\bigskip

\textbf{Acknowledgments.} 
I would like to thank A.~Acin, J.~Bowles, D.~Jennings, K.~Korzekwa, R.~Kunjwal, M.~Oszmaniec and M.~Perarnau-Llobet for insightful discussions and comments on a previous draft. I acknowledge financial support from the Spanish MINECO (Severo Ochoa
SEV-2015-0522 and project QIBEQI FIS2016-80773-P), Fundacio Cellex, Generalitat de Catalunya (CERCA
Programme and SGR 875) and COST Action MP1209.
\bibliography{Bibliography_thermodynamics}

\onecolumngrid

\section{Supplemental Material}

\subsection{Ontological models and genuine non-classicality}
\label{appendix:ontological}

The working definition we used of the notion of  ``genuinely non-classical'' in the main text is ``a phenomenon that cannot be reproduced within any non-contextual local ontological model''. We now review a way in which these concepts can be formalised. See also Refs.~\cite{mermin1993hidden, spekkens2005contextuality,ferrie2011quasi, spekkens2014status, leifer2014quantum, brunner2014bell, kunjwal2015kochen} for extended discussions.

\subsubsection{Operational theory}

First of all, one has the statistics collected from experiments. 
Formally this defines an \emph{operational theory}, a set of operational notions of preparation procedures $P$ and measurement procedures $M$, together with a function
\begin{equation}
	(P,M) \longmapsto p(k|P,M),
\end{equation}
associating to each couple $(P,M)$ the statistics generated by the measurement $M$ on the preparation $P$. $k$ labels the outcomes of $M$.

The operational theory is also provided with notions of convex combinations $\sum_i q_i P^i$ and $\sum_i q_i M^i$, that correspond to choosing the preparation $P_i$ or the measurement $M_i$ from the outcome of a classical random variable distributed according to $\{q_i\}$; moreover, a notion of coarse-graining of a measurement $M$ is introduced. If $K$ is the set of outcomes of $M$, we can partition $K$ into sets $K_1$,..,$K_n$, obtaining a new measurement $\widetilde{M}$ with outcomes $K_1$,...,$K_n$. Operationally, $\widetilde{M}$ is realised by performing $M$ and recording outcome $j$ whenever $k \in K_j$.

\subsubsection{Ontological model}
\label{sec:ontologicalmodel}

Secondly, one looks for what we called in the main text a ``mechanism'' reproducing the statistics. This has been formalised in the literature through the notion of \emph{ontological model} \cite{spekkens2005contextuality, leifer2014quantum}, and much before that by the notion of \emph{hidden variable model} \cite{bell1966problem}. Given an operational theory, an ontological model poses the existence of a set of physical states $\lambda$ in some measure space $\Lambda$ such that
\begin{enumerate}
	\item For every preparation $P$ there exists a probability $p(\lambda|P)$ over $\Lambda$. This models the fact that every time one follows the preparation procedure $P$, state $\lambda$ is prepared with probability $p(\lambda|P)$.
	\item For every measurement $M$ there exists a probability $p(k|\lambda,M)$, with $p(k|\lambda,M) \in [0,1]$ and $\sum_k p(k|\lambda,M) = 1$ for every $\lambda \in \Lambda$. This models the response function of the measurement $M$, i.e. the probability that the measurement procedure $M$ returns outcome $k$, given that the physical state is $\lambda$.
\end{enumerate}

Finally, the ontological model is required to be compatible with the above notions of convex combinations and coarse-graining, i.e. \cite{oszmaniec2017private}
\begin{equation}
	\label{eq:ontologicalconvexity}
	\sum_i q_i P^i \mapsto \sum_i q_i p(\lambda|P^i), \quad   \sum_i q_i M^i \mapsto \sum_i q_i p(k|\lambda,M^i),
\end{equation}
\begin{equation}
	\label{eq:ontologicalcoarsegraining}
	p(K_j|\lambda,\widetilde{M}) =  \sum_{k \in K_j} p(k|\lambda,M).
\end{equation}

\subsubsection{Non-locality and contextuality}
The following discussion is based upon Refs.~\cite{spekkens2005contextuality, spekkens2014status}.

One may be tempted to define non-classical any phenomenon that cannot be captured within any ontological model or ``mechanism'', but this idea needs to be sharpened. In fact, without further restrictions, \emph{any} quantum statistics can be reproduced~\cite{spekkens2005contextuality}. Consider the Beltrametti-Bugajski model, in which $\Lambda$ is given by the set of projectors $\Psi$ \cite{beltrametti1995classical}. The model associates to the preparation of a pure state $\Psi$ the delta function
\begin{equation}
	p(\lambda|\Psi) = \delta (\lambda- \Psi).
\end{equation}
Moreover, every POVM $\{M_k\}$ is associated to the response function
\begin{equation}
	p(k|\lambda, \{M_k\}) = \tr{}{M_k \ketbra{\lambda}{\lambda} }.
\end{equation}
Using the convexity assumption, it is simple to show that this model reproduces the Born rule for every preparation and measurement. One then needs to specialise the notion of ``classical mechanism'' in the main text with some extra restrictions on the ontological model. We briefly summarise here locality and contextuality. 

Famously, in the case of two space-like separated measurements $\{M_k\}$ and  $\{N_j\}$, one wishes to make a \emph{locality} assumption. This requires that, once $\lambda$ is given, the response function should factorise, namely
\begin{equation*}
	p(k,j|\lambda, \{M_k \otimes N_j\}) =  p(k|\lambda, \{M_k\})p(j|\lambda, \{N_j\}).
\end{equation*}
Under the assumption that the measurements can be chosen independently of $\lambda$ and of each other, the impossibility of any such local ontological model to reproduce the quantum predictions is called \emph{Bell non-locality} \cite{brunner2014bell}, a standard notion of non-classicality.

Non-contextuality, defined in Sec.~B of the main text assuming the operational theory provided by quantum mechanics, can be viewed in much the same way. It is an elementary assumption, holding true in classical theory, that cannot be maintained in any ontological model reproducing the quantum predictions. The generalisation of the definitions given in the main text reads as follows. An ontological model is said to be \emph{universally non-contextual} if any two operationally indistinguishable preparations (and measurements) are represented in the same way in the hidden variable theory \footnote{We do not delve here into the notion of transformation non-contextuality. For an extended discussion, see Ref.~\cite{spekkens2005contextuality}}. Specifically, \emph{preparation non-contextuality} is defined as follows: if $p(k|P,M) = p(k|P',M)$ for every $M$ and $k$ ($P$ and $P'$ are operationally indistinguishable), then $p(\lambda|P)= p(\lambda| P')$ for all $\lambda$ ($P$ and $P'$ are the same preparation in terms of $\lambda$'s). \emph{Measurement non-contextuality} is defined in the same way, inverting the roles of $P$ and $M$:  if $p(k|P,M) = p(k|P,M')$ for every $P$ and $k$ ($M$ and $M'$ are operationally indistinguishable), then $p(k|\lambda,M)= p(k|\lambda, M')$ for all $k$, $\lambda$. The failure of any non-contextual ontological model to explain the observed statistics is termed \emph{contextuality}~\cite{spekkens2005contextuality}. We note in passing that the Beltrametti-Bugajski model could be ruled out as a classical mechanism on the grounds that it violates preparation non-contextuality. 

One can make a comparison with the locality assumption~\cite{spekkens2005contextuality}: 
\begin{enumerate}
	\item We are given an operational notion of no-signalling: no experiment ever managed to signal faster than light.
	\item The natural explanation of this fact is that the hidden variables themselves are not signalling.
\end{enumerate}

Finally, we comment on the connection between the definitions of contextuality given in the main text and Kochen-Specker contextuality. 
Consider the assumption of \emph{outcome determinism}: sharp measurements $M$ are associated to indicator functions, i.e., $p(k|\lambda, M) \in \{0,1\}$ for all $k$ and $\lambda \in \Lambda$. Under this assumption, a measurement non-contextual model is non-contextual in the sense of Kochen and Specker \cite{kochen1975problem, kunjwal2015kochen}. In fact, take the operational theory to be given by quantum theory. Using the defining properties of ontological models one can show that measurement non-contextuality implies $p(k|\lambda, M) =p(k|\lambda, M_k)$, where $M_k$ is the POVM element associated to outcome $k$ of measurement $M$  \cite{spekkens2014status}. If $M$ is associated to a hermitian operator, $M_k$ is just an element of the basis of eigenvectors of $M$ and, from outcome determinism, $p(k|\lambda, M_k) \in \{0,1\}$. We could now consider the joint measurement of two commuting observables $A$ and $B$ or two commuting observables $A$ and $C$. Then, given $\lambda$, the above reasoning shows that a $\{0,1\}$ assignment must be made for every projector of $A$, independently of the choice of the other elements of the basis, i.e. independently of the fact that $A$ is measured jointly with $B$ or with $C$ (the ``context''). This is, however,  Kochen-Specker non-contextuality.

\subsection{General protocol for strong and weak measurements of work}
\label{appendix:weak}

We discuss here in more detail the family of protocols described in the main text and the two limits $s \rightarrow 0$ and $s \rightarrow \infty$ (see Fig.~\ref{fig:weakscheme}). We focus on the case in which for every $w$ there is a unique couple of indexes $i$ and $j$ such that $w = E'_j - E_i$. 

\begin{figure}
	\centering
	\includegraphics[width=0.4\linewidth]{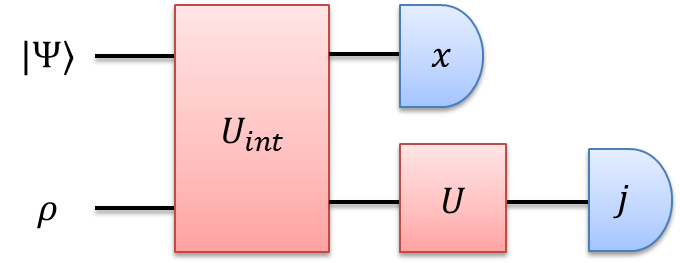}
	\caption{A one-parameter family of protocols for the measurement of work in a system initially prepared in state $\rho$ and undergoing a unitary evolution $U$. The protocols are characterised by the preparation of the pointer state $\Psi$, a Gaussian state with width $s$. The pointer interacts through $U_{\rm int}$ with the system and it is then projectively measured in the position basis, returning outcome $x$. The system, on the other hand, after interacting with the pointer evolves according to the driving unitary $U$ and its final energy is projectively measured, returning outcome $j$. As we describe in detail, in the limit $s \rightarrow 0$ one recovers the two-point measurement distribution, whereas for $s \rightarrow \infty$ we obtain a weak measurement scheme.}
	\label{fig:weakscheme}
\end{figure}

The initial state of system and measurement device is $\rho \otimes \Psi$, where $\Psi = \ketbra{\Psi}{\Psi}$ and $\ket{\Psi}$ is given by
\begin{equation}
	\label{eq:ancilla}
	\ket{\Psi} = \int dx G_{s}(x) \ket{x}, \quad G_{s}(x)= (\pi s^2)^{-1/4} \exp \left[-x^2/(2s^2)\right].
\end{equation}
The unitary interaction induced by $H_{\rm int}$, coupling system and pointer, reads
\begin{equation}
	\label{eq:unitaryinteraction}
	U_{\rm int} := e^{-i \mathcal{E}_i \otimes P} = \mathcal{E}_i \otimes e^{-i P} + \mathcal{E}^\perp_i \otimes \iden,
\end{equation}
where $\mathcal{E}^\perp_i$ is the projector on the subspace orthogonal to $\mathcal{E}_i = \ketbra{i}{i}$. Let us introduce the notation
\begin{equation*}
	\rho^{11} = \mathcal{E}_i \rho \mathcal{E}_i,  \; \rho^{01} = \mathcal{E}^\perp_i \rho \mathcal{E}_i, \; \rho^{10} = \mathcal{E}_i \rho \mathcal{E}^\perp_i, \; \rho^{00} = \mathcal{E}^\perp_i \rho \mathcal{E}^\perp_i.
\end{equation*}
The state of system and measurement device after the interaction~is
\begin{small}
	\begin{equation*}
		U_{\rm int} \rho \otimes \Psi U^\dag_{\rm int} = \sum_{k, k' =0}^1 \rho^{kk'} \otimes \int_{-\infty}^{+\infty} dx dy G_{s}(x -k) G_{s}(y-k') \ketbra{x}{y}.
	\end{equation*}
\end{small}
We perform a projective measurement $\{\ketbra{x}{x}\}$ on the device. The next steps are the unitary driving $U$ on the system, followed by an energy measurement with respect to the final Hamiltonian. The expectation value of the device position, upon postselecting outcome $j$ in the final energy measurement, is denoted by $\langle X \rangle_j$. It reads
\begin{equation}
	\label{eq:Xj}
	\begin{split}
		\langle X \rangle_j &= \tr{}{\left(\ketbra{\tilde{j}}{\tilde{j}} \otimes X\right) U_{\rm int} \rho \otimes \Psi U^\dag_{\rm int}}/q_j \\
		&:= \tr{}{ X \sigma_j}/q_j,
	\end{split}
\end{equation}
where $\ket{\tilde{j}} = U^\dag \ket{j}$, \mbox{$\sigma_j := \bra{\tilde{j}} U_{\rm int} \rho \otimes \Psi U^\dag_{\rm int} \ket{\tilde{j}}$} and $q_j$ is the probability of observing outcome $j$ in the final energy measurement. We have
\begin{equation*}
	\sigma_j = \sum_{k,k'=0}^1 \bra{\tilde{j}}\rho^{kk'}\ket{\tilde{j}}\int_{-\infty}^{+\infty} dx dy G_{s}(x -k) G_{s}(y-k') \ketbra{x}{y}.
\end{equation*}
The probability $q_j= \tr{}{\sigma_j}$ is given by
\begin{eqnarray*}
	q_j &= &\bra{\tilde{j}}\rho^{00}\ket{\tilde{j}} + \bra{\tilde{j}}\rho^{11}\ket{\tilde{j}} + 2e^{-1/(4s^2)} \Re \bra{\tilde{j}}\rho^{10} \ket{\tilde{j}} \\
	&=& \bra{\tilde{j}}\rho\ket{\tilde{j}} - 2(1-e^{-1/(4s^2)}) \Re \bra{\tilde{j}}\rho^{10} \ket{\tilde{j}},
\end{eqnarray*}
where for the second line we used 
\begin{equation*}
	\bra{\tilde{j}}\rho\ket{\tilde{j}} = \bra{\tilde{j}}\rho^{00}\ket{\tilde{j}} + \bra{\tilde{j}}\rho^{11}\ket{\tilde{j}} + 2 \Re \bra{\tilde{j}}\rho^{10}\ket{\tilde{j}}.
\end{equation*}
Moreover,
\begin{equation*}
	\tr{}{X \sigma_j} = \bra{\tilde{j}}\rho^{11}\ket{\tilde{j}} + e^{-1/(4s^2)} \Re \bra{\tilde{j}}\rho^{10} \ket{\tilde{j}}.
\end{equation*}
From Eq.~\eqref{eq:Xj}, one has the final expression
\begin{equation}
	\label{eq:conditionalexpectation}
	\langle X \rangle _j = \frac{\bra{\tilde{j}}\rho^{11}\ket{\tilde{j}} + e^{-1/(4s^2)} \Re \bra{\tilde{j}}\rho^{10} \ket{\tilde{j}}}{\bra{\tilde{j}}\rho\ket{\tilde{j}} - 2(1-e^{-1/(4s^2)}) \Re \bra{\tilde{j}}\rho^{10} \ket{\tilde{j}}}.
\end{equation}
We can now study the two relevant limits of Eq.~\eqref{eq:conditionalexpectation}:
\begin{enumerate}
	\item When $s \rightarrow 0$, also called the \emph{strong measurement} limit, one has 
	\begin{equation}
		\langle X \rangle _j \rightarrow \frac{\bra{\tilde{j}}\rho^{11}\ket{\tilde{j}}}{q_j} = \frac{p_{\rm tpm}(w|\mathcal{P})}{q_j}.
	\end{equation}
	\item Consider now the limit $s \rightarrow \infty$, also called the \emph{weak measurement} limit. Using $\mathcal{E}_i + \tilde{\mathcal{E}}_i = \mathbb{1}$ and defining $\Pi_j = \ketbra{\tilde{j}}{\tilde{j}}$,
	\begin{equation}
		\label{eq:conditionalexpectationweak}
		\langle X \rangle _j \rightarrow \frac{\bra{\tilde{j}}\rho^{11}\ket{\tilde{j}} + \Re \bra{\tilde{j}}\rho^{10}\ket{\tilde{j}} }{\bra{\tilde{j}}\rho\ket{\tilde{j}}} = \frac{\Re \tr{}{\rho \mathcal{E}_i \Pi_j}}{\tr{}{\rho \Pi_j}} = \frac{p_{\rm weak}(w|\mathcal{P})}{q_j}.
	\end{equation}
\end{enumerate}
The theoretical meaning of the quantity $p_{\rm weak}(w|\mathcal{P})$ can also be understood from its connection to an inference procedure in the presence of non-commuting observables \cite{hall2004prior}. 

The quantity in Eq.~\eqref{eq:conditionalexpectationweak} can also be identified with the real part of the (generalised) \emph{weak value} of $\mathcal{E}_i$, with pre-selection $\rho$ and post-selection $\Pi_j$. For an introduction, see \cite{dressel2014colloquium, aharonov1988result, dressel2010contextual, hosoya2010strange}. A weak value is called anomalous when its value lies outside the spectrum of the observable. This can happen in two ways: either the weak value has non-zero imaginary component, or its real part lies outside the spectrum of the observable (or both). Here we focus on the real part of the weak value, because 1. The quasi-probability of work introduced by Allahverdyan in Ref.~\cite{allahverdyan2014nonequilibrium} is proportional to the real part of the weak value and 2. Proofs of contextuality rely on an anomaly in the real part (see Ref.~\cite{pusey2014anomalous, pusey2015logical}), while a non-zero imaginary part of the weak values can be reproduced within non-contextual models~\cite{karanjai2015weak}. Hence, for our purposes we will follow Ref.~\cite{pusey2014anomalous} and call \emph{anomalous} a weak value with anomalous real part. Since $\sum_i \Re \tr{}{\rho \mathcal{E}_i \Pi_j}/\tr{}{\rho \Pi_j} =1$, it follows without loss of generality that if one of the $\mathcal{E}_i$ is anomalous in its real part, there will be some $\mathcal{E}_j$ such that the weak value has negative real part. For this reason we can focus on the question of what are the consequences of a negative value in $p_{\rm weak}(w|\mathcal{P})$. 

\subsection{Some subtleties concerning Theorem~1}

\begin{enumerate}
	\item \emph{What is the connection between the present study and Ref.~\cite{spekkens2008negativity}?} 
	
	We first need to recall what precisely is the claim of Ref.~\cite{spekkens2008negativity}. There it is shown that given a prepare and measure scenario, where one prepares a set of quantum states $\rho^j$ and picks measurements from a set of POVMs $\{M^l_k\}$ with outcomes $k$, the existence of a universally non-contextual model for the ensuing statistics $\tr{}{\rho^j M^l_k}$ is equivalent to the existence of a non-negative quasi-probability representation (a positive representation, for short) for the relevant preparations and measurements. A positive representation is given by a measure space of ontological states $\lambda \in \Lambda$ and the association of $\rho^j$, $M^l_k$ to normalised, non-negative functions on $\Lambda$:
	\begin{equation}
		\label{eq:representations}
		\rho^j \mapsto p(\lambda|\rho^j), \quad M^l_k \mapsto p(k|\lambda, M^l_k).
	\end{equation} 
	Here the functions $p(\lambda|\rho^j)$, $p(k|\lambda, M^l_k)$ are convex respectively in $\rho_j$ and $M^l_k$ and they need to satisfy
	\begin{equation*}
		\tr{}{\rho^j M^l_k}= \int d\lambda p(\lambda|\rho^j) p(k|\lambda, M^l_k).
	\end{equation*}
	The existence of such positive representation is then seen to be equivalent (by definition) to universal non-contextuality. 
	
	The conclusion one can draw is that the negativity of $p(w|\mathcal{P})$ is a very different concept compared to the negativity of quasi-probability representations. As the examples $p(w|\mathcal{P}) = \tr{}{\rho M_w}$ and $p(w|\mathcal{P}) = \Re \tr{}{\rho \mathcal{E}_i \Pi_j}$ from the main text show, when we define $p(w|\mathcal{P})$ we are not even associating functions separately to states and measurements. In other words, $p(w|\mathcal{P})$ does not define a quasi-probability \emph{representation}, positive or otherwise. Hence, the negativity of $p(w|\mathcal{P})$ is not negativity of the functions $p(\lambda|\rho^j)$
	and $p(k|\lambda,M_k^l)$ over ontic states, and $p(w|\mathcal{P})$ is not required to reproduce the observed statistics by averaging over some $\lambda$. 
	
	In short, one should carefully distinguish various notions of negativity used in the literature. Showing that in certain cases the negativity of some quasi-probability implies the negativity of every quasi-probability \emph{representation} of the protocol $\mathcal{P}$ is a highly non-trivial task, first accomplished by Pusey in the paper \cite{pusey2014anomalous}.

	\item \emph{Can one derive Theorem~1 independently of the assumption that $\mathcal{P}$ is a FT protocol?} One can notice that assumption 1 by itself simply implies that $p(w|\mathcal{P})$ is associated to some POVM ${M_w}$, so that from Eq.~(5) in the main text
	\begin{equation}
		\tr{}{\rho M_w} = \int_\Lambda d\lambda p(\lambda|\rho) p(w|\lambda,M_w(\mathcal{P})).
	\end{equation}
	The question is then if the statistics on the left-hand side can be reproduced by a non-contextual ontological model for a given set of preparations $\rho$ and measurements $M_w$. Without the assumption that $M_w$ reproduces the TPM scheme for classical states, $M_w$ is a completely arbitrary set of measurements. Hence, the answer to the above question is negative, since there are measure and prepare schemes exhibiting contextuality~\cite{spekkens2005contextuality}. This shows that the possibility of constructing a non-contextual model is granted by the assumption of recovering the TPM scheme for classical states together with assumption~1. Assumption~1 alone is insufficient. 
\end{enumerate}  

\subsection{Proof of Theorem~2}

With reference to the protocol introduced in Sec.~D of the main text, and dropping for simplicity unnecessary indexes $i$ and $j$, we report here Theorem~2:
\begin{thm}
	\label{thm:contextuality}
	Let $\rho$ be a quantum state, $\mathcal{E}$ and $\Pi$ projectors. If $\Re \tr{}{\rho \mathcal{E} \Pi} <0$, for $s$ large enough there is no measurement non-contextual ontological model for preparation $\rho$, measurement $\{M^s_x\}_{x \in \R}$ and post-selection $\Pi$ that satisfies outcome determinism.     
\end{thm}

We break down this theorem into two independent lemmas and use some of the notation of Ref.~\cite{pusey2014anomalous} when possible to ease comparisons. The first technical lemma shows that $\Re \tr{}{\rho \mathcal{E} \Pi} <0$ implies a generalised version of the assumptions of Theorem~1 in Ref.~\cite{pusey2014anomalous}:

\begin{lem}
	\label{lemma}
	Let $\rho$ be a quantum state and $\mathcal{E}$, $\Pi$ be projectors. Assume $\Re \tr{}{\rho \mathcal{E} \Pi} <0$. Then
	\begin{enumerate}
		\item \label{cond1} $p_{\Pi} := \tr{}{\Pi \rho}>0$. 
		\item \label{cond2} The family of POVMs $M^s_x= N^{s \dag}_x N^s_x$ satisfies \begin{enumerate}
			\item \label{cond2a} $M^s_x = p^s(x-1) \mathcal{E} + p^s(x) \tilde{\mathcal{E}}, \; \; \tilde{\mathcal{E}} = \mathbb{1} - \mathcal{E}, \; \;$ $p^s(x)$ probability distribution with median $x=0$.
			\item \label{cond2b} $S^s:= \int_{-\infty}^{+\infty} N^{s \dag}_x \Pi N^s_x dx = (1-p^s_d) \Pi + p^s_d E_d$, with $\{ E_d, \mathbb{1} - E_d \}$ POVM and $p^s_d \in [0,1/2]$.
			\item \label{cond2c} For $s$ large enough,  $p^s_-:= \frac{1}{p_\Pi} \int_{-\infty}^0 \tr{}{N^{s \dag}_x \Pi N^s_x \rho} dx > \frac{1}{2} + \frac{p^s_d}{p_\Pi}$. 
		\end{enumerate} 
	\end{enumerate}
\end{lem}
\begin{proof}
	From Cauchy-Schwarz,
	\begin{equation*}
		0<\left| \tr{}{\rho \mathcal{E} \Pi} \right|^2  \leq \tr{}{\Pi \rho} \tr{}{\rho \mathcal{E}}.
	\end{equation*}
	One must then have $\tr{}{\Pi \rho} > 0$. This shows that \ref{cond1} holds.
	
	As in Supplemental Material~\ref{appendix:weak}, let us define an ancillary state $\ket{\Psi}$ as in Eq.~\eqref{eq:ancilla} and a unitary interaction as in Eq.~\eqref{eq:unitaryinteraction}. By measuring $\{\ketbra{x}{x}\}$ on the ancilla, we define the POVM $M^s_x = N^{s \dag}_x N^s_x$ on the system. Since
	\begin{equation*}
		U_{\rm int} \ket{\Psi} = \mathcal{E} \otimes \int_{-\infty}^{+\infty} dx G_s(x)\ket{x+1} + \mathcal{\tilde{E}} \otimes  \int_{-\infty}^{+\infty} dx G_s(x)\ket{x} = \mathcal{E} \otimes \int_{-\infty}^{+\infty} dx G_s(x-1)\ket{x} + \mathcal{\tilde{E}} \otimes  \int_{-\infty}^{+\infty} dx G_s(x)\ket{x}
	\end{equation*}
	one has
	\begin{equation}
		\label{eq:povmNx}
		N^s_x = \bra{x} U_{\rm int} \ket{\Psi} = G_s(x-1) \mathcal{E} + G_s(x) \tilde{\mathcal{E}},
	\end{equation} 
	\begin{equation}
		M^s_x:= N^{s \dag}_x N^s_x = G_s^2 (x-1) \mathcal{E} + G_s^2(x) \tilde{\mathcal{E}}.
	\end{equation}
	We can then recognise that $p^s(x):=  G_s^2(x)$ has median $x=0$, from which we obtain condition \ref{cond2a}. To verify the other conditions, the following integrals will be useful:
	\begin{equation}
		\int_{-\infty}^{+\infty} G_s(x-a) G_s (x-b) dx = \exp[-(a-b)^2/(4s^2)].
	\end{equation}
	Substituting Eq.~\eqref{eq:povmNx} in the definition of $S$ and using the integrals above one then has
	\begin{equation}
		S^s = \mathcal{E} \Pi \mathcal{E} + \tilde{\mathcal{E}} \Pi \tilde{\mathcal{E}} + e^{-\frac{1}{4s^2}}(\mathcal{E} \Pi \tilde{\mathcal{E}} + \tilde{\mathcal{E}} \Pi \mathcal{E}) .
	\end{equation}
	From $\mathcal{E} + \tilde{\mathcal{E}} = \mathbb{1}$,
	\begin{equation*}
		\mathcal{E} \Pi \mathcal{E} + \tilde{\mathcal{E}} \Pi \tilde{\mathcal{E}} = \frac{1}{2}\Pi +  \frac{1}{2}(	\mathcal{E} \Pi \mathcal{E} + \tilde{\mathcal{E}} \Pi \tilde{\mathcal{E}})  - \frac{1}{2}(	\mathcal{E} \Pi \tilde{\mathcal{E}} + \tilde{\mathcal{E}} \Pi \mathcal{E}) = \frac{1}{2}\Pi + \frac{1}{2} (\mathcal{E} - \tilde{\mathcal{E}})\Pi(\mathcal{E} - \tilde{\mathcal{E}}) 
	\end{equation*}
	\begin{equation*}
		\mathcal{E} \Pi \tilde{\mathcal{E}} + \tilde{\mathcal{E}} \Pi \mathcal{E} = \frac{1}{2}\Pi +  \frac{1}{2}(	\mathcal{E} \Pi \tilde{\mathcal{E}} + \tilde{\mathcal{E}} \Pi \mathcal{E})  - \frac{1}{2}(	\mathcal{E} \Pi \mathcal{E} + \tilde{\mathcal{E}} \Pi \tilde{\mathcal{E}}) = \frac{1}{2}\Pi - \frac{1}{2} (\mathcal{E} - \tilde{\mathcal{E}})\Pi(\mathcal{E} - \tilde{\mathcal{E}}) 
	\end{equation*}
	Substituting these relations in the expression for $S^s$, and defining $E_d =  (\mathcal{E} - \tilde{\mathcal{E}})\Pi(\mathcal{E} - \tilde{\mathcal{E}})$,
	\begin{equation}
		S^s = \frac{1+ e^{-1/(4s^2)}}{2}\Pi + \frac{1- e^{-1/(4s^2)}}{2} E_d
	\end{equation}
	Since $E_d ^2 = E_d$, $\{E_d, \mathbb{1} - E_d\}$ defines a POVM. Setting $p^s_d = \frac{1- e^{-1/(4s^2)}}{2}$, we obtain condition \ref{cond2b}.
	
	Finally, for the last condition the following integrals will be useful (erfc denotes the complementary Gauss error function, i.e. $ {\rm erfc}(s) = 1- {\rm erf}(s)$, where erf is the Gauss error function):
	\begin{equation*}
		\int_{-\infty}^{0}G_s^2(x-1)dx = \frac{1}{2} {\rm erfc}\left(\frac{1}{s}\right)   , \quad \int_{-\infty}^{0}G_s^2(x)dx = \frac{1}{2}, \quad \int_{-\infty}^{0}G_s(x-1) G_s(x)dx = \frac{e^{-1/(4s^2)}}{2} {\rm erfc}\left(\frac{1}{2s}\right)
	\end{equation*}
	Let us then compute $p^s_-$, substituting Eq.~\eqref{eq:povmNx} and using the above integrals:
	\begin{equation}
		p^s_- = \frac{1}{p_\Pi} \left\{ \frac{1}{2} {\rm erfc}\left(\frac{1}{s}\right)    \tr{}{\mathcal{E} \Pi \mathcal{E} \rho}  + \frac{e^{-1/(4s^2)}}{2}  {\rm erfc}\left(\frac{1}{2s}\right) \tr{}{(\tilde{\mathcal{E}} \Pi \mathcal{E} + \mathcal{E} \Pi \tilde{\mathcal{E}})\rho} + \frac{1}{2}\tr{}{\tilde{\mathcal{E}} \Pi \tilde{\mathcal{E}} \rho}  \right\}.
	\end{equation}
	Expanding around $s\rightarrow +\infty$,
	\begin{equation*}
		\frac{1}{2} {\rm erfc}\left(\frac{1}{s}\right)  = \frac{1}{2} - \frac{1}{\sqrt{\pi}s} + o\left(\frac{1}{s}\right), \quad  \frac{1}{2}e^{-1/(4s^2)} {\rm erfc}\left(\frac{1}{2s}\right)  = \frac{1}{2} - \frac{1}{2\sqrt{\pi}s} + o\left(\frac{1}{s}\right).
	\end{equation*}
	From this it follows
	\begin{eqnarray*}
		p^s_- &=& \frac{1}{2} - \frac{1}{ 2p_\Pi \sqrt{\pi}s} \tr{}{(\mathcal{E} \Pi \mathcal{E} + \tilde{\mathcal{E}} \Pi \mathcal{E})\rho}  - \frac{1}{ 2 p_\Pi \sqrt{\pi}s} \tr{}{(\mathcal{E} \Pi \mathcal{E} + \mathcal{E} \Pi \tilde{\mathcal{E}})\rho} + o\left(\frac{1}{s}\right)\\
		& = & \frac{1}{2} - \frac{1}{ 2p_\Pi \sqrt{\pi}s} \tr{}{\Pi \mathcal{E} \rho + \mathcal{E} \Pi \rho} = \frac{1}{2} - \frac{1}{ p_\Pi \sqrt{\pi}s} \Re \tr{}{\rho \mathcal{E} \Pi} + o\left(\frac{1}{s}\right)	\end{eqnarray*}
	Since $p^s_d = 1/(8s^2) + o\left(\frac{1}{s^2}\right) = o\left(\frac{1}{s}\right)  $,
	\begin{equation}
		p^s_- - \frac{1}{2} - \frac{p^s_d}{p_\Pi} = -\frac{1}{ p_\Pi \sqrt{\pi}s} \Re \tr{}{\rho \mathcal{E} \Pi} + o\left(\frac{1}{s}\right)
	\end{equation}
	Hence, given that $\Re \tr{}{\rho \mathcal{E} \Pi} < 0$, for $s> 0$ large enough condition \ref{cond2c} is satisfied.
	
\end{proof}

The second lemma is the analogue of the main theorem of Pusey in Ref.~\cite{pusey2014anomalous} under the generalised assumptions derived through the previous lemma:

\begin{lem}
	Let $\rho$ be a quantum state and $\mathcal{E}$, $\Pi$ be projectors. Under conditions \ref{cond1} and \ref{cond2} of Lemma~\ref{lemma}, there exists no measurement non-contextual ontological model for preparation $\rho$, measurement $\{M^s_x\}_{x \in \R}$ and postselection of $\Pi$ satisfying outcome determinism. 
\end{lem}
\begin{proof}
	The argument of Ref.~\cite{pusey2014anomalous} carries through in our generalised situation, as we now show. To ease comparisons, we use a notation similar to that introduced in Ref.~\cite{pusey2014anomalous}.
	
	Consider the measurement $\{M^s_x = N^{s \dag}_x N^s_x\}$ followed by $\{\Pi,1-\Pi\}$ as a single measurement $\{S^s_x\} \cup \{F^s_x\}$, where $S^s_x = N^{s \dag}_x \Pi N^s_x$, $F^s_x = N^{s \dag }_x (1- \Pi) N^s_x$ correspond to successful or failed post-selection, respectively. The measurement non-contextual ontological model is required to satisfy
	\begin{equation}
		\label{eq:basicrequirement}
		\tr{}{\Pi N^s_x \rho N^{s \dag}_x} = \tr{}{S^s_x \rho}= \int_{\Lambda} p(x|S^s_x,\lambda) p(\lambda|\rho)d\lambda.
	\end{equation}
	The idea of the proof is 
	\begin{enumerate}
		\item  Find bounds on $p(x|S^s_x,\lambda)$ that every measurement non-contextual and outcome determinism model must satisfy if conditions \ref{cond1}, \ref{cond2a} and \ref{cond2b} hold.
		\item Show that these bounds are incompatible with condition \ref{cond2c}.
		\item Since conditions \ref{cond1} and \ref{cond2a}, \ref{cond2b}, \ref{cond2c} hold by assumption, we conclude that at least one among measurement non-contextuality and outcome determinism must go.
	\end{enumerate}
	We will use the assumption of measurement non-contextuality throughout the proof. Note that $M^s_x$ is a coarse-graining of $S^s_x$ and $F^s_x$, in the sense that $M^s_x = S^s_x + F^s_x$. From the coarse-graining assumption in the definition of ontological model (Eq.~\eqref{eq:ontologicalcoarsegraining}), we hence must have 
	\begin{equation}
		\label{eq:slessthane}
		p(x|M^s_x,\lambda) = p(x|S^s_x,\lambda) + p(x|F^s_x,\lambda) \quad \Rightarrow \quad  p(x|S^s_x,\lambda) \leq p(x|M^s_x,\lambda)
	\end{equation} 
	We now use condition~\ref{cond2a} and the convexity assumption of ontological models (Eq.~\eqref{eq:ontologicalconvexity}). Denoting by $1$ and $0$, respectively, the outcomes associated to projectors $\mathcal{E}$ and $\tilde{\mathcal{E}}$, we must have
	\begin{equation}
		\label{eq:decomposee}
		p(x|M^s_x,\lambda) = p^s(x-1) p(1|\mathcal{E},\lambda) + p^s(x)p(0|\tilde{\mathcal{E}},\lambda)
	\end{equation}
	The median of $p^s(x)$ is zero due to condition~\ref{cond2a}, so 
	\begin{equation}
		\label{eq:medianbound}
		\int_{-\infty}^{0} p^s(x-1)dx \leq \int_{-\infty}^{0} p^s(x)dx = 1/2
	\end{equation}. Hence, combining Eqs.~\eqref{eq:slessthane}-\eqref{eq:medianbound} we obtain the first of the two bounds we need
	\begin{equation}
		\label{eq:boundSx1}
		\int_{-\infty}^0 p(x|S^s_x,\lambda) dx\leq \int_{-\infty}^0 p(x|M^s_x,\lambda) dx \leq \frac{1}{2}\left( p(1|\mathcal{E},\lambda) + p(0|\tilde{\mathcal{E}},\lambda) \right) = \frac{1}{2}.
	\end{equation}
	Now recall that $S^s = \int_{-\infty}^{+\infty} S^s_x$ and consider the POVM $\{S^s,\mathbb{1} - S^s\}$ with outcomes $1$, $0$, respectively. From condition~\ref{cond2b} and using again the coarse-graining and convexity assumptions of ontological models of Eqs.~\eqref{eq:ontologicalconvexity}-\eqref{eq:ontologicalcoarsegraining}
	\begin{equation}
		\label{eq:S}
		S^s = (1-p^s_d)\Pi + p^s_d E_d \quad \Rightarrow \quad \int_{-\infty}^{+\infty} p(x|S^s_x, \lambda) dx = (1-p^s_d) p(1|\Pi,\lambda) + p^s_d p(1|E_d,\lambda).
	\end{equation}
	Note that we denoted by $1$, $0$ the outcomes of $\{E_d, \mathbb{1} - E_d\}$, respectively. 
	
	Now, due to outcome determinism and the fact that $\{\Pi, \mathbb{1} - \Pi\}$ is a projective measurement, we can partition $\Lambda$ in two disjoint sets $\Lambda_0$ and $\Lambda_1$, where
	\begin{equation}
		\label{eq:probabilityofQ}
		p(1|\Pi,\lambda) = \left\{ 
		\begin{array}{ll} 
			0 \quad & \lambda \in \Lambda_0 \\
			1 \quad & \lambda \in \Lambda_1
		\end{array} \right.
	\end{equation}
	Then, for all $\lambda \in \Lambda_0$, Eq.~\eqref{eq:S} implies
	\begin{equation}
		\label{eq:boundSx2}
		\int_{-\infty}^{0}p(x|S^s_x,\lambda) dx\leq \int_{-\infty}^{+\infty}p(x|S^s_x,\lambda)dx =  p^s_d p(1|E_d,\lambda) \leq p^s_d.
	\end{equation}
	This is the second bound we were looking for. Given Eq.~\eqref{eq:boundSx1} and Eq.~\eqref{eq:boundSx2}, we are ready to derive a contradiction with condition~\ref{cond2c}. Due to Eq.~\eqref{eq:basicrequirement}, we have
	\begin{equation}
		p^s_-:= \frac{1}{p_\Pi} \int_{-\infty}^0 \tr{}{N^{s \dag}_x \Pi N^s_x \rho} dx = \frac{1}{p_\Pi} \int_{-\infty}^0 \int_{\Lambda_0} p(x|S^s_x,\lambda) p(\lambda|\rho) dx d\lambda + \frac{1}{p_\Pi} \int_{-\infty}^0 \int_{\Lambda_1} p(x|S^s_x,\lambda) p(\lambda|\rho) dx d\lambda
	\end{equation}
	In the first term, let us use Eq.~\eqref{eq:boundSx2}, while in the second we use the bound of Eq.~\eqref{eq:boundSx1}. This gives, 
	\begin{equation}
		p^s_- \leq \frac{p^s_d}{p_\Pi} \int_{\Lambda_0}  p(\lambda|\rho) d\lambda + \frac{1}{2p_\Pi} \int_{\Lambda_1}  p(\lambda|\rho) d\lambda \leq \frac{p^s_d}{p_\Pi} + \frac{1}{2p_\Pi} \int_{\Lambda_1}  p(\lambda|\rho) d\lambda,
	\end{equation}
	where for the second inequality we used $\int_{\Lambda_0}  p(\lambda|\rho) d\lambda \leq 1$, which follows from the normalisation condition.
	Eq.~\eqref{eq:probabilityofQ} and the definition of ontological models lead to the following chain of equalities:
	\begin{equation}
		\frac{1}{2p_\Pi} \int_{\Lambda_1}  p(\lambda|\rho) d\lambda = \frac{1}{2p_\Pi} \int_{\Lambda_1}  p(1|\Pi,\lambda) p(\lambda|\rho) d\lambda = \frac{1}{2p_\Pi} \int_{\Lambda}  p(1|\Pi,\lambda) p(\lambda|\rho) d\lambda = \frac{1}{2p_\Pi} \tr{}{\Pi\rho} = \frac{1}{2}
	\end{equation}
	Hence we conclude
	\begin{equation}
		p^s_- \leq \frac{p^s_d}{p_\Pi} + \frac{1}{2},
	\end{equation}
	Since this is in contradiction with condition~\ref{cond2c}, we conclude.
\end{proof}

\end{document}